\documentclass[runningheads,envcountsame]{llncs}
\usepackage{amssymb, amsmath}

\usepackage[polish,english]{babel}
\usepackage[T1,nomathsymbols]{polski}
\usepackage[utf8]{inputenc}

\usepackage{stmaryrd}
\usepackage{mathtools}
\usepackage[geometry]{ifsym}
\usepackage{paralist}
\usepackage{makecell}
\usepackage{color, colortbl}
\usepackage{multirow}
\usepackage{multicol}
\usepackage{mathrsfs}

\usepackage{cleveref}

\usepackage[ruled,vlined,linesnumbered]{algorithm2e}
\Crefname{algocf}{Algorithm}{Algorithms}

\usepackage{pgfplots}
\usepackage{tikz}
\pgfplotsset{width=10cm,compat=1.9}
\usetikzlibrary{arrows,backgrounds,automata,topaths,patterns,decorations.pathreplacing,decorations.pathmorphing}
\usetikzlibrary{arrows, chains, positioning, quotes, shapes.geometric}
\usetikzlibrary{calc}
\usetikzlibrary{fit}
\usetikzlibrary{arrows}
\usetikzlibrary{positioning,automata}

\newcommand{\HL}{\ensuremath{\mathcal{H}}}

\newcommand{\MLC}{\ensuremath{\mathcal{MLC}}}
\newcommand{\ML}{\ensuremath{\mathcal{ML}}}

\newcommand{\MLi}{\ensuremath{\mathcal{ML}(\DD)}}

\newcommand{\cl}{\ensuremath{\mathsf{cl}}}
\newcommand{\iotaf}{\ensuremath{\DD}}

\newcommand{\md}{\ensuremath{\mathsf{md}}}

\newcommand{\D}{\ensuremath{\mathsf{D}}}
\newcommand{\E}{\ensuremath{\mathsf{E}}}
\newcommand{\A}{\ensuremath{\mathsf{A}}}

\newcommand{\F}{\mathcal{F}}

\newcommand{\M}{\mathcal{M}}
\newcommand{\MN}{\mathcal{N}}

\newcommand{\Z}{\mathbb{Z}}
\newcommand{\N}{\mathbb{N}}

\newcommand{\Lan}{\mathcal{L}}
\newcommand{\FR}{\mathsf{F}}
\newcommand{\Hs}{\mathcal{H}}
\newcommand{\C}{\mathit{Current}}

\newcommand{\s}{s}

\newcommand{\Pow}{\mathcal{P}}
\newcommand{\prop}{\mathsf{PROP}}
\newcommand{\nom}{\mathsf{NOM}}

\newcommand{\DD}{\ensuremath{\mathsf{DD}}}

\newcommand{\PS}{\ensuremath{\textsc{\textup{PSpace}}}}

\newcommand{\EXPT}{\ensuremath{\textsc{\textup{ExpTime}}}}
\newcommand{\NEXPT}{\ensuremath{\textsc{\textup{NExpTime}}}}

\usepackage{todonotes}

\usepackage[T1]{fontenc}
\usepackage{graphicx}
\begin{document}
\title{
Hybrid Modal Operators for Definite Descriptions}
%
%
\author{
Przemysław Andrzej Wałęga\inst{1,2}\orcidID{0000-0003-2922-0472}
\and
Michał Zawidzki\inst{1,2}\orcidID{0000-0002-2394-6056}
}

\authorrunning{P. A. Wałęga et al.}
%
\institute{University of Łódź, Poland 
\\
\and
University of Oxford, United Kingdom\\
\email{\{przemyslaw.walega,michal.zawidzki\}@cs.ox.ac.uk}
}
\maketitle              
\begin{abstract}
In this paper, we study computational complexity and expressive power 
of modal operators for definite descriptions, which correspond to statements `the modal world  which satisfies formula $\varphi$'.
We show that 
adding such operators to
the basic (propositional) modal language
has a price of increasing complexity of the satisfiability problem from \PS{} to \EXPT{}.
However, if formulas corresponding to descriptions are Boolean only, there is no increase of complexity.
Furthermore, we compare definite descriptions   with the related operators from hybrid  and counting logics.
We prove that the operators for definite descriptions are  strictly more expressive than hybrid operators, but  strictly less expressive than counting operators.
We show that over linear structures  the same expressive power results hold as in the general case;
in contrast, if the linear structures are isomorphic to integers, 
definite descriptions become as expressive as counting operators.
%
\end{abstract}

\section{Introduction}

\emph{Definite descriptions}  are term-forming expressions  such as 
`the $x$ such that $\varphi(x)$', 
which are usually represented 
with 
Peano's   $\iota$-operator
as 
$\iota x \varphi(x)$~\cite{peano1897studii}.
Such expressions intend to denote a single object satisfying a property $\varphi$,
but
providing a complete formal theory for them turns out to be a complex task
due to several non-intuitive
cases, for example, when there exists  no object satisfying  $\varphi$, 
when there are multiple such objects,
or when a formula with a definite description is in the scope of negation. 
As a result,  a number of competing 
theories have been proposed~\cite{HilBer68,Rosser78,Scott67,Lambert2001,benc:free86},
including Russell's famous 
approach according to which 
the underlying logical form of a sentence 
`$\iota x \varphi(x)$ satisfies $\psi$'
is that 
`there exists exactly one $x$ which satisfies $\varphi$ and 
moreover this $x$ satisfies  $\psi$'~\cite{PelLin2005}.

More recently it has been observed that definite descriptions, and referring expressions in general,
provide a convenient way of identifying objects  in information and knowledge base management systems~\cite{borgida2016referring,artale2021free}.
Such expressions can be used to replace obscure identifiers~\cite{borgida2016referring,borgida2017concerning},
enhance query answering~\cite{toman2019finding},
identify problems in conceptual modelling~\cite{borgida2016referring2},
and
identity resolution in ontology-based data access~\cite{toman2018identity,toman2019identity}.
For this reason referring expressions have been studied in the setting of description logics (DLs)
\cite{areces2008referring,toman2019finding,Mazzullo2022DD}---well-known formalisms for ontologies and the Semantic Web.
In particular, Neuhaus et al.
\cite{neuhaus2020free}
introduced free DLs (free of the presupposition that each term denotes) with three alternative 
dual-domain semantics:  positive, negative, and gapping, where statements in ABoxes and TBoxes involving non-referring expressions can still be true, become automatically false, or lack a truth value, respectively.
Artale et al. \cite{artale2021free}, in turn, 
proposed free DLs using single domain semantics;
they introduced
definite descriptions in DLs by 
allowing for expressions of the form $\{ \iota C \}$, whose extension is a singleton containing the unique element of which a (potentially complex) concept $C$ holds, or the empty set if there does not exist  such a unique element.
%
%
Definite descriptions can therefore be seen as a generalisation of nominals, which in DLs take the form $\{a\}$ with $a$ being an individual name. Since Artale et al. do not assume that all individual names must refer, a nominal $\{a\}$ with $a$ being a non-referring name, denotes the empty set.
As shown by  Artale et al.  \cite{artale2021free},
definite descriptions can be simulated  in DLs with nominals and the universal role.
In particular, adding definite descriptions to $\mathcal{ALCO}_u$ (i.e., $\mathcal{ALC}$ with nominals and the universal role) does not increase the computational complexity of checking ontology satisfiability, which remains
 \EXPT{}-complete.


In modal logics nominals are treated as specific atoms which must hold in single modal worlds \cite{prior1967past,bull1970approach,blackburn1993nominal,gargov1993modal}.
Satisfaction operators $@_i$, in turn, are indexed with nominals $i$ and allow us to write formulas such as $@_i \varphi$, whose meaning is that $\varphi$ holds in the unique modal world in which nominal $i$ holds (but $\varphi$ can also hold in other worlds).
Nominals and satisfaction operators constitute the standard \emph{hybrid machinery}, which added to the basic modal logic gives rise to the hybrid logic $\HL(@)$~\cite{areces200714,Blackburn2000}.
Such a machinery increases the expressiveness of the basic modal logic by making it possible, for example, to  encode irreflexivity or atransitivity of the accessibility relation.
At the same time the computational complexity of the satisfiability problem in $\HL(@)$ remains  \PS{}-complete, so the same as in the basic modal logic~\cite{areces1999road}. 
On the other hand, introducing further  hybrid operators or considering temporal hybrid logics oftentimes has a drastic impact on the computational complexity~\cite{areces200714,ten2005complexity,areces1999road,areces2000computational,franceschet2003hybrid}.

Closely related are also the \emph{difference} $\D$ and the \emph{universal} $\A$ modalities.
Adding any of them to the basic modal language makes the satisfiability problem \EXPT{}-complete~\cite{blackburn2002modal}.
It is  not hard to show that $\D$ allows us to express nominals and satisfaction operators;
what is more interesting, however, is that the basic modal logic with $\D$ is equivalent to the hybrid modal logic with $\A$~\cite{gargov1993modal}.
Furthermore, one can observe that
having access to both $\A$ and nominals enables to  express definite descriptions by marking with a nominal the unique world in which the definite description holds, and using $\A$ to state that this description holds only it the world satisfying this nominal (as observed by Artale et al.~\cite{artale2021free}).

Uniqueness of a world  can also be expressed in the \emph{modal logic with counting} \MLC{}, which extends the basic modal language with counting operators of the form $\exists_{\geq n}$, where  $n \in \N$, and $\exists_{\geq n} \varphi$ states that $\varphi$ holds in at least $n$ distinct worlds~\cite{areces2010modal}.
Using Boolean connectives and $\exists_{\geq n}$ enables to also express the counting operators of the forms $\exists_{\leq n}$ and $\exists_{=n}$.
Such operators can be used to encode the hybrid machinery, as well as $\A$ and $\D$, but this comes at a considerable complexity cost. 
In particular, the satisfiability problem in \MLC{} is \EXPT{}-complete if numbers $n$ in counting operators are encoded in unary~\cite{tobies2001complexity} and it is \NEXPT{}-complete if the numbers are encoded in binary~\cite{Zawidzki2013,pratt2005complexity,pratt2010two}.

In contrast to the extensive studies of hybrid and counting modal operators, as well as definite descriptions in first-order modal logics \cite{fitting2012first,indrzejczak2018cut,orlandelli2018labelled,indrzejczak2020existence},
definite descriptions have not been thoroughly analysed in  propositional modal logics, which we address in this paper.
To this end, we consider the basic modal language and  extend it with a (hybrid) modal operator for definite descriptions $@_{ \varphi}$ which can be indexed with an arbitrary modal formula $\varphi$. 
The intuitive meaning of $@_\varphi \psi$ is that $\psi$ holds in the unique world in which $\varphi$ holds.
Our goal is to determine the computational cost of adding such definite descriptions to the language, and to investigate the expressive power of the obtained logic, denoted as \MLi{}.

The main contributions of this paper are as follows:
\begin{enumerate}
\item We show that 
adding to the basic modal language definite descriptions $@_{\varphi}$ with Boolean $\varphi$  (so $\varphi$ does not mention modal operators) can be done with no extra computational cost.
In other words, satisfiability of \MLi{}-formulas with Boolean definite descriptions is \PS{}-complete.
The main part of the proof is to show the upper bound by reducing \MLi{}-satisfiability to the existence of a winning strategy in a specific game played on Hintikka sets.

\item On the other hand, if we allow for arbitrary $\varphi$'s in definite descriptions, the satisfiability problem becomes \EXPT{}-complete.
Thus, the computational price of adding non-Boolean definite descriptions is the same as for adding the universal modal operator $\A$ or counting operators $\exists_{\geq n}$ with numbers $n$ encoded in unary.
The important ingredient of the  proof is showing the lower bound by reducing satisfiability in the basic modal logic with the universal modality $\A$ to \MLi{}-satisfiability.

\item We show that, over the class of all frames, \MLi{} is strictly more expressive than $\HL(@)$, but strictly less expressive than \MLC{}.
In particular, 
\MLC{} can define frames with domains of cardinality $n$, for any $n\in\N$. On the other hand, the only frame cardinality 
\MLi{} can define is $1$,
and $\HL(@)$ cannot define any frame properties related to cardinality.

\item We prove that over linear frames the same expressiveness results hold as for arbitrary frames, but over the integer frame \MLi{} becomes as expressive as \MLC{}.
In particular, over such a frame the operators $\exists_{\geq n}$ become expressible in \MLi{}, which is still not the case for $\HL(@)$.
\end{enumerate}
The rest of the paper is organised as follows.
In \Cref{logic} we present \MLi{} formally.
We obtain its syntax by extending the basic modal logic with definite description operators $@_{\varphi}$ and we provide the semantics for these operators exploiting the standard Russellian theory of definite descriptions.
We also present $\HL(@)$ and \MLC{}, which are considered in the later parts of the paper.
In \Cref{complexity} we prove both of our computational complexity results, namely tight \PS{} and \EXPT{} bounds.
Then, in \Cref{expressiveness} we turn our attention to expressive power;
we define notions used to compare the expressive power of the logics in question and present a variant of bisimulation which is adequate for \MLi{}.
We show results that hold over arbitrary and linear frames, and we finish with results that hold over integers.
Finally, we briefly conclude the paper in \Cref{conclusions}.

\section{Logic of Definite Descriptions and Related Formalisms}\label{logic}

In what follows, we  introduce formally the modal logic of definite descriptions \MLi{}
and present closely related logics which were studied in the literature.

We let formulas of
$\MLi$ be defined as in the basic modal logic, but we additionally allow for using the operator $@$ to construct formulas of the form $@_{ \varphi} \psi$ whose intended meaning is that formula $\psi$ holds in the unique world in which formula $\varphi$ holds.

Formally, $\MLi$-formulas 
are generated by the  grammar
\begin{align*}
\varphi & \Coloneqq
p \mid
\neg \varphi \mid
\varphi \lor \varphi \mid
\Diamond \varphi \mid
@_{ \varphi} \varphi ,
\end{align*}
where $p$ ranges over the set $\prop$ of propositional variables.
We refer to an expression $@_{ \varphi}$ as a \emph{definite description}---\DD{} in short---and we call it \emph{Boolean} 
if so is $\varphi$ (i.e., $\varphi$ does not mention $\Diamond$ or $@$).
We will also use $\bot$, $\top$, $\land$, $\to$, and $\Box$, which stand for the usual abbreviations.
We let $\prop(\varphi)$ be the set of propositional variables occurring in $\varphi$ and 
the \emph{modal depth}, $\md(\varphi)$, of $\varphi$  the 
deepest nesting of $\Diamond$ in $\varphi$.

We will consider the Kripke-style semantics of \MLi{}, where a \emph{frame} is a pair ${\F = (W, R)}$ consisting of a non-empty set  $W$ of worlds and 
an accessibility relation $R \subseteq W \times W$. 
A  \emph{model} based
on a frame ${ \F= (W, R) }$ is a tuple ${\M = (W, R, V )}$, 
where
$V : \prop \longrightarrow \Pow(W)$ is a valuation assigning a set of worlds to each propositional
variable.
The \emph{satisfaction relation} $\models$ for  $\M = (W, R, V )$ and 
$w \in W$ is defined inductively as follows:
\begin{align*}
& \M,  w  \models  p  && \text{iff} &&   w \in V(p), \text{ for each }  p \in \prop
\\
& \M,  w  \models  \neg  \varphi  && \text{iff} &&  \M,  w   \not \models   \varphi
\\
& \M,  w  \models \varphi_1 \lor \varphi_2  && \text{iff} && \M,  w  \models   \varphi_1 \text{ or } \M,  w  \models    \varphi_2
\\
& \M,  w  \models \Diamond  \varphi && \text{iff} && \text{there exists } v \in W \text{ such that } (w,v) \in R \text{ and }  \M,  v  \models   \varphi
\\
& \M,  w  \models @_ { \varphi_1}  \varphi_2  && \text{iff} && \text{there exists } v \in W \text{ such that } \M,  v  \models   \varphi_1  , \M,  v  \models   \varphi_2  
\\ 
& && && \text{and  } \M,  v'  \not\models   \varphi_1 \text{ for all } v' \neq v \text{ in } W
\end{align*}
We say that $\varphi$ is \emph{satisfiable} if there exist $\M$ and $w$ such that $\M, w \models \varphi$;
we will focus on checking satisfiability as the main reasoning task.

It is worth observing that \MLi{} allows us to naturally express definite descriptions with both the \emph{external} and \emph{internal} negation.
The first type of negation corresponds to sentences of the form `it is not the case that the $x$ such that $\varphi$ satisfies $\psi$' which can be written as $\neg @_{ \varphi} \psi$.
The internal negation occurs in sentences of the form `the $x$ such that $\varphi$ does not satisfy $\psi$', which can be expressed in \MLi{} as
$@_{ \varphi} \neg \psi$.



Next, we  present well-studied extensions of the basic modal language
which are particularly relevant for investigating \MLi{}, namely the  logic $\MLC$ with counting operators $\exists_{\geq n}$, 
with any $n \in \N$~\cite{areces2010modal,areces2000computational}, 
and the logic $\HL(@)$ with hybrid  operators $@_i$, where $i$ is a nominal (i.e., an atom which holds in exactly one world)~\cite{areces200714,areces1999road}.
The intended reading of
$\exists_{\geq n } \varphi$ is that $\varphi$ holds in at least $n$ distinct worlds, whereas
$@_i \varphi$ is that $\varphi$ holds in the unique world labelled by $i$.

Formally, $\MLC$-formulas are generated by the grammar
\begin{align*}
\varphi & \Coloneqq
p \mid
\neg \varphi \mid
\varphi \lor \varphi \mid
\Diamond \varphi \mid
\exists_{\geq n} \varphi,
\end{align*}
where $p \in \prop$ and $n \in \N$.
We will also use  $\exists_{\leq n} \varphi$ as an abbreviation for $\neg \exists_{\geq n +1} \varphi$ and $\exists_{=n} \varphi$ as an abbreviation for $\exists_{\geq n} \varphi \land \exists_{\leq n} \varphi$.
The semantics of \MLC{} is obtained by extending the basic modal logic semantics with the condition
\begin{align*}
& \M,  w  \models \exists_{\geq n}  \varphi  && \text{iff} && \text{there are at least } n \text{ worlds } v\in W \text{ such that } \M,  v  \models   \varphi
\end{align*}
Formulas of $\HL(@)$, in turn, are generated by the grammar
\begin{align*}
\varphi & \Coloneqq
p \mid
i \mid
\neg \varphi \mid
\varphi \lor \varphi \mid
\Diamond \varphi \mid
@_i \varphi,
\end{align*}
for $p \in \prop$ and $i$ belonging to the set $\nom$ of \emph{nominals}.
The semantics of $\HL(@)$ exploits
\emph{hybrid models}   $\M = (W, R, V )$ which are defined like standard modal models except that
$V : \prop \cup \nom \longrightarrow \Pow(W)$  assigns not only sets of worlds to  propositional
variables, but also singleton sets to nominals.
Then the conditions of the satisfaction relation are extended with
\begin{align*}
& \M,  w  \models  i  && \text{iff} &&  V(i) = \{w\}, \text{ for each } i \in \nom
\\
& \M,  w  \models @_i  \varphi && \text{iff} && \M,  v  \models   \varphi, \text{ for } v  \text{ such that }V(i) = \{v\}
\end{align*}

We can already observe some relations between definite descriptions $@_{  \varphi}$, the counting operator $\exists_{=1}$, and satisfaction operators $@_i$.
For example, $@_{  \varphi} \psi$ can be expressed as $\exists_{=1} \varphi \land \exists_{=1} (\varphi \land \psi)$, which states that $\varphi$ holds in a single world and that $\psi$ also holds in this world.
On the other hand we can simulate a nominal $i$ with a propositional variable $p_i$ by writing a formula $@_{ p_i} \top$, which guarantees the existence of the unique world in which $p_i$ holds.
Then $@_i \varphi$ can be simulated as $@_{ p_i} \varphi$;
note that for the latter simulation we use only Boolean \DD{}s. 
In the following sections we will study the relation between logics with these operators in detail.
In particular, we will aim to determine how the complexity and expressiveness of \MLi{} compares to the ones of the related logics.


\section{Computational Complexity}\label{complexity}

In this section, we investigate the computational complexity of the satisfiabily problem in \MLi{}.
First, we show that if we allow for Boolean \DD{}s only, the problem is \PS{}-complete, that is, the same as in the language without \DD{}s;
hence,  extending the language in this way can be performed with no computational cost.
However, in the second result we show that in the case of arbitrary \DD{}s the problem becomes \EXPT{}-complete, and so, the computational price of adding \DD{}s is the same as for adding 
counting quantifiers (with numbers encoded in unary)~\cite{tobies2001complexity,areces2000computational} or for adding the universal modality~\cite{blackburn2002modal}.

We start by showing \PS{}-completeness of the satisfiability problem in the case of Boolean \DD{}s.
The lower bound follows trivially from \PS{}-completeness of the same problem in basic modal logic~\cite{ladner1977computational,blackburn2002modal}.
For the upper bound, we show that the problem reduces to checking the existence of a winning strategy in a specific two-player game.
States of this game can be represented in polynomial space, and so, we can check the existence of a winning strategy in \PS{}.
It is worth observing that a similar technique was used to show the \PS{} upper bound for $\HL(@)$~\cite{areces1999road} and for
modal logics of topological spaces with the universal modality~\cite{sustretov2007topological}.



Our game for  checking if an input formula $\varphi$ is satisfiable will be played using $\varphi$-Hintikka sets defined as follows.

\begin{definition}
We let the \emph{closure}, $\cl(\varphi)$, of an $\MLi$-formula $\varphi$ be the minimal set of formulas which contains all subformulas of $\varphi$, and such that if $\psi \in \cl(\varphi)$ but $\psi$ is not of the form   $\neg \chi$, then $\neg \psi \in \cl(\varphi)$. 
A
\emph{$\varphi$-Hintikka set} $H$
is any maximal subset of $\cl(\varphi)$ which satisfies  the following conditions, for all $\psi,\psi_1, \psi_2 \in H$:
\begin{itemize}
\item if  $\neg \psi \in \cl(\varphi)$, then $\neg \psi \in H$ if and only if $\psi \not\in H$,

\item if  $\psi_1 \lor \psi_2 \in \cl(\varphi)$, then $\psi_1 \lor \psi_2 \in H$ if and only if $\psi_1 \in H$ or $\psi_2 \in H$.
\end{itemize}

\end{definition}
For example, if $\varphi$ is of the form 
$@_{  \neg ( p \lor \neg p) }$
then
$\{ p, (p \lor \neg p),  @_{  \neg ( p \lor \neg p) } \}$ constitutes a $\varphi$-Hintikka set.
Note that although $\varphi$-Hintikka sets are consistent with respect to Boolean connectives, they do not need to be consistent (i.e., satisfiable) in general;
indeed,  $@_{  \neg ( p \lor \neg p) }$ in the set above is unsatisfiable.

Given the definition of a $\varphi$-Hintikka set we are ready to present the game.
To this end, we will use the symbol $\iotaf(\varphi)$ to represent the set
 of all formulas $\psi$ such that $@_{ \psi}$ occurs in $\varphi$.

\begin{definition}
For an $\MLi$-formula $\varphi$ we let
the $\varphi$-game 
be played between Eloise and Abelard as follows.
In the first turn Eloise needs to provide a set
$\Hs$
of at most $|\iotaf(\varphi) | +1$ $\varphi$-Hintikka sets and  
a relation $R \subseteq \Hs \times \Hs$
such that:
\begin{itemize}
\item 
$\varphi \in H$, for some $H \in \Hs$,

\item each $\psi \in \iotaf(\varphi)$ can occur in at most one $H \in \Hs$,

\item for all $@_{ \psi} \chi \in \cl(\varphi)$ and  $H \in \Hs$  we have $@_{ \psi} \chi \in H$ iff there is $H' \in \Hs$ such that $\{ \psi , \chi \}  \subseteq H'$,

\item and for all $\Diamond \psi \in \cl(\varphi)$,
if $R(H,H')$  and $\psi \in H'$, then $\Diamond \psi \in H$.
\end{itemize}
Then Abelard and Eloise play in turns.
Abelard selects  $H \in \C$ (initially $\C = \Hs$) and  a formula $\Diamond \varphi' \in H$, which he wants to verify.
This $\Diamond \varphi'$ needs to have the modal depth  not larger than $\md(\varphi)$ decreased by the number of turns Abelard already played.
Then it is Eloise's turn in which she needs to provide a witnessing $\varphi$-Hintikka set 
$H'$ such that
\begin{itemize}
\item $\varphi' \in H'$,

\item if $H' \cap \iotaf(\varphi) \neq \emptyset$, then  $H' \in \Hs$,

\item for all $@_{ \psi} \chi \in \cl(\varphi)$ we have $@_{ \psi} \chi \in H'$ iff there is $H'' \in \Hs $ such that $\{ \psi , \chi \}  \subseteq H''$,

\item and for all $\Diamond \psi \in \cl(\varphi)$, if $\psi \in H'$, then $\Diamond \psi \in H$.
\end{itemize}
If $H' \cap \iotaf(\varphi) \neq \emptyset$, then Eloise wins.
Otherwise the game continues with Abelard's  turn in which $H'$ is added to $\Hs$ and  the set $\C$ becomes $\{ H' \}$.
When one of the players cannot make any move, the game ends and this player loses.
\end{definition}

We observe that a $\varphi$-game needs to terminate, as Abelard can play at most $\md(\varphi)+1$ turns.
Moreover, we show next that
verifying the satisfiability of $\varphi$ reduces to checking the existence of Eloise's winning strategy in the $\varphi$-game.

\begin{lemma}\label{lem:game}
For any $\MLi$-formula $\varphi$ with Boolean \DD{}s, $\varphi$ is satisfiable if and only if  Eloise has a winning strategy in the $\varphi$-game.
\end{lemma}
\begin{proof}
If $\varphi$ is satisfiable, then Eloise can construct a winning strategy by reading the required $\varphi$-Hintikka sets from a model of $\varphi$.
For the opposite direction, assume that Eloise has a winning strategy that starts by playing $\Hs_0 = \{H_0, \dots, H_n \}$.
We define  $\Hs_1, \dots, \Hs_{\md(\varphi)}$ such that each $\Hs_{k+1}$ is the set of
all $\varphi$-Hintikka sets not belonging to $\Hs_0$ which Eloise would play (using the winning strategy)
as a response to Abelard having played some set (and a formula) in $\Hs_k$.
We exploit these $\Hs_0, \dots, \Hs_{\md(\varphi)}$ to construct a model $\M=(W,R,V)$  such that 
\begin{align*}
W & = \{ w_k^H \mid k \in \{0, \dots, \md(\varphi) \} \text{ and } H \in \Hs_k \},
\\
R & = \{ (w_k^H, w_{k'}^{H'}) \in W \times W  \mid \psi \in {H'} \text{ implies } \Diamond \psi \in H, \text{ for all } \Diamond \psi \in \cl(\varphi) \},
\\
V(p) & = \{w_k^H \in W \mid p \in H \}, \quad \text{ for each } p \in \prop.
\end{align*}
We can show by induction on the structure of formulas
that for any 
$w_k^H \in W$
and any 
$\psi \in \cl(\varphi)$  with $\md(\psi) \leq \md(\varphi) -k $  it holds that $\M, w_k^H \models \psi$ if and only if $\psi \in H$.
Thus, $\M, w_0^{H}\models \varphi$, for $H \in \Hs_0$ such that $\varphi \in H$ (which needs to exist by the definition of the $\varphi$-game). 
\qed
\end{proof}

We observe that each state of the $\varphi$-game can be represented in polynomial space with respect to the size of $\varphi$.
In particular, in each state we need to specify a set of polynomially many $\varphi$-Hintikka sets played so far, each containing polynomially many formulas, which in total uses polynomial space.
The existence of a winning strategy for Eloise can therefore be decided in \PS{} (e.g., by exploiting the fact that \PS{} coincides with the class of problems decided by alternating Turing machines in polynomial time~\cite{chandra1981alternation}).

\begin{theorem}\label{pspace}
Checking satisfiability of \MLi{}-formulas with Boolean \DD{}s is \PS{}-complete.
\end{theorem}

Importantly, \Cref{pspace} does not hold if we allow for non-Boolean \DD{}s, which disallows us to conduct the induction from the proof of
\Cref{lem:game}. 
As we show next, this is not a coincidence, namely the satisfiability problem for \MLi{} with non-Boolean \DD{}s is \EXPT{}-complete.

The \EXPT{} upper bound follows from an observation that \DD{}s can be simulated with the counting operator $\exists_{=1}$;
recall that we can simulate $@_{ \varphi} \psi$ with $\exists_{=1} \varphi \land \exists_{=1} (\varphi \land \psi)$.
As we use only one counting operator $\exists_{=1}$ and \MLC{}-satisfiability with numbers encoded in unary is \EXPT{}-complete~\cite{tobies2001complexity,areces2000computational}, our upper bound follows.
The proof of the matching  lower bound
is more complicated and is obtained by simulating the universal modal operator 
$\A$ with \DD{}s, where $\A \varphi$ stands for  `$\varphi$ holds in all worlds'.
To simulate $\A$  we start by
guaranteeing that there exists a unique `trash'  world in which a special propositional variable $\s$ holds and which is accessible with $\Diamond$ only from itself; this can be obtained by the formula $@_{\s}\top \land @_{ \Diamond \s} \s$.
Now, we can use this world to simulate $\A \varphi$ with
$@_{ (\s \lor \neg \varphi) }\top$, which states that $\varphi$ holds in all worlds in which $\s$ does not hold, that is, in all worlds different from our `trash' world.
Although this does not allow us to express the exact meaning of $\A \varphi$, it turns out to be sufficient to reduce satisfiability of  formulas of the logic $\ML(\A)$ with the $\A$ operator to \MLi{}-satisfiability.
As the former problem is \EXPT{}-complete~\cite{blackburn2002modal}, we obtain the required lower bound.


\begin{theorem}\label{expt}
Checking satisfiability of \MLi{}-formulas (with arbitrarily complex \DD{}s) is \EXPT{}-complete.
\end{theorem}
\begin{proof}
As we have observed, the upper bound is trivial, so we focus on showing \EXPT{}-hardness.
To this end, we reduce  $\ML(\A)$-satisfiability to \MLi{}-satisfiability.
First, given an $\ML(\A)$-formula, we transform it into a formula $\varphi$ in the negation normal form NNF, where negations occur only in front of propositional variables.
This can be done in logarithmic space, but requires using additional operators, namely $\land$, $\Box$, and $\E$.
In particular, $\E$ stands for `somewhere' and is dual to $\A$ similarly to $\Diamond$ being dual to $\Box$.
Then, we construct a translation of such formulas in NNF  to \MLi{}-formulas as follows:
\begin{align*}
\tau(p) & =  p , && & \tau(\Diamond\psi) & =  \Diamond\tau(\psi),
\\
\tau(\neg p) &  =  \neg p, && & \tau(\Box\psi) & =  \Box\tau(\psi),
\\
\tau(\psi \lor \chi) & = \tau(\psi) \lor \tau(\chi), && & \tau(\E \psi) & = @_{ p_\psi} ( \tau(\psi) \land \neg \s),
\\
\tau(\psi \land \chi) & = \tau(\psi) \land \tau(\chi), && & \tau(\A\psi) & =  @_{ (\s \lor \neg \tau(\psi)) }\top, 
\end{align*}
where $p \in \prop$, $\psi$ and $\chi$ are subformulas of $\varphi$, $\s$ is a fresh variable marking a `trash' world, and $p_\psi$ is a fresh variable for each $\psi$. 
Our finally constructed formula $\varphi'$ is defined as follows:
$$
\varphi ' = 
\tau(\varphi) \land \neg \s \land @_{ \s}\top \land @_{ \Diamond \s} \s.
$$ 
Since $\varphi'$ is constructed in logarithmic space from $\varphi$, it remains to show that $\varphi$ and $\varphi'$ are equisatisfiable.

If $\varphi$ is satisfiable, then $\M , w \models \varphi$, for some $\M= (W,R,V)$ and $w\in W$.
To show that $\varphi'$ is satisfiable, we construct, in two steps,
a model $\M' = (W',R',V')$ extending $\M$.
First, for each subformula $\psi$ of $\varphi$ which is satisfied in some world in $\M$ we choose an arbitrary
world $v \in W$ such that
$\M,v \models \psi $ and we let ${ V'(p_\psi) = \{ v \} }$.
Second, we add a single new world $w_\s$ to $W'$ as well as we set  ${V'(\s) = \{ w_\s \} }$ and $(w_\s,w_\s) \in R'$.
Then, we can show by induction on the structure of $\varphi$ that for all  $v\in W$,
if $\M, v \models \varphi$
then $\M', v \models \tau(\varphi)$.
This, in particular, implies that $\M',w \models \tau(\varphi)$. 
By the construction of $\M'$ we have also $\M',w \models \neg \s \land @_{\s}\top\land @_{ \Diamond \s} \s$, so we can conclude that $\M',w \models \varphi'$.

For the opposite direction we assume that $\varphi'$ is satisfiable, so  $\M',w \models \varphi'$ for some $\M'=(W',R',V')$ and $w \in W'$.
In particular
$\M',w \models \neg \s \land @_{ \s}\top \land @_{ \Diamond \s} \s$, so
there exists a unique world $w_\s \in W'$ such that $\M', w_\s \models \s$, and $\M',w \models \neg \s$ implies that  $w_s \neq w$.
Now, we construct $\M=(W,R,V)$ by deleting from $\M'$ the world $w_s$ and restricting the accessibility relation and the valuation to this smaller set of worlds.
Then, we can show by induction on the structure of $\varphi$ that for any  $v \in W$, if $\M ', v \models \tau(\varphi)$, then $\M, v \models \varphi$.
Since $\M',w \models \tau(\varphi)$ and $w \in W$, we obtain that $\M, w \models \varphi$.
\qed
\end{proof}

Note that the reduction in the proof above  provides us  with a satisfiability preserving translation between languages.
The existence of such a reduction does not mean, however, that
there exists a translation preserving
equivalence of formulas.
In the next section we will study the existence of the second type of translations to compare the expressiveness  of \MLi{} with that of $\HL(@)$ and~\MLC{}.

\section{Expressive Power}\label{expressiveness}

In the previous section we have established the computational complexity of reasoning in \MLi{}. 
Now, we will 
compare \MLi{} with  $\HL(@)$ and \MLC{} from the point of view of expressiveness.
We will study their relative expressive power  over the class of all frames, as well as  over linear frames $L$ (where the accessibility relation  is irreflexive, transitive, and trichotomous), and over the frames $\Z$ which are isomorphic to the standard (strict) order of integers.

To this end, for a class $\FR$ of frames below we define the \emph{greater-than expressiveness} relation $\preccurlyeq_\FR$ (we drop the index  $\FR$ in the case of all frames).
If logics $\Lan_1$ and $\Lan_2$ are  non-hybrid, then we let 
$\Lan_1 \preccurlyeq_\FR \Lan_2$, if,
for any 
$\Lan_1$-formula $\varphi$, there is an $\Lan_2$-formula $\varphi'$ 
such that
$\M,w \models \varphi$ if and only if
$\M,w \models \varphi'$, for any model $\M$ based on a frame from the class $\FR$ and any world $w$ in $\M$.
If $\Lan_1$ is hybrid but $\Lan_2$ is not,  
we treat nominals 
as fresh propositional variables in $\Lan_2$, so
we can still require that $\M,w \models \varphi$ implies 
$\M,w \models \varphi'$.
For the opposite direction we require that if $\M,w \models \varphi'$, for a non-hybrid model $\M=(W,R,V)$, then $V(i)$ is a singleton for each $i \in \nom(\varphi)$;
thus we can treat $\M$ as a hybrid model and require now that $\M,w \models \varphi$.
If $\Lan_1$ is non-hybrid but $\Lan_2$ is hybrid, we 
define 
$\Lan_1 \preccurlyeq \Lan_2$ analogously.
Then,  $\Lan_2$ has a \emph{strictly higher expressiveness} than $\Lan_1$, in symbols $\Lan_1 \prec_\FR \Lan_2$, if 
$\Lan_1 \preccurlyeq_\FR \Lan_2$, but $\Lan_2 \not\preccurlyeq_\FR \Lan_1$,
whereas   $\Lan_1$ have the \emph{same  expressiveness} as $\Lan_2$, in symbols $\Lan_1 \approx_\FR \Lan_2$, if both $\Lan_1 \preccurlyeq_\FR \Lan_2$ and $\Lan_2 \preccurlyeq_\FR \Lan_1$.


For $\Lan_1 \preccurlyeq_\FR  \Lan_2$ it  suffices to construct a translation, but
showing that ${\Lan_1 \not\preccurlyeq_\FR  \Lan_2}$ is usually more complicated. It can be obtained, for example, by using an adequate notion of bisimulation, which we present for \MLi{} below.

\begin{definition}\label{def::bisim}
A $\DD$-\emph{bisimulation} between  $\M=(W,R,V)$ and $\M'=(W',R',V')$ is any total (i.e., serial and surjective) relation ${Z \subseteq W \times W'}$ such that whenever $(w,w') \in Z$, the following conditions hold:
\begin{description}
\item[Atom:]
$w$ and $w'$ satisfy the same propositional variables, 

\item[Zig:]
if there is $v \in W$ such that $(w,v) \in R$, then there is $v' \in W'$ such $(v,v') \in Z$ and $(w',v') \in R'$,

\item[Zag:]
if there is $v' \in W'$ such that $(w',v') \in R'$, then there is $v \in W$ such $(v,v') \in Z$ and $(w,v) \in R$,

\item[Singular:]
$Z(w)= \{ w' \}$  if and only if  $Z^{-1}(w') = \{ w \}$\footnote{We use here the functional notation where $Z(w) = \{v \mid (w,v) \in Z \}$.}.
\end{description}
\end{definition}

Note that by relaxing the definition of $\DD$-bisimulation, namely 
not requiring the totality of $Z$ and removing  Condition (Singular), we obtain the standard notion of bisimulation, which is adequate for basic modal language~\cite{blackburn2002modal,blackburn2006handbook}.
Additional restrictions imposed on the bisimulation give rise to bisimulations adequate for $\MLC{}$ and $\HL(@)$.
In particular,  \MLC-bisimulation is defined by extending the standard bisimulation (for basic modal language) with the requirement that $Z$ contains a bijection between $W$ and $W'$ \cite{areces2010modal}.
In turn, 
an $\HL$-bisimulation introduces to the standard bisimualtion an additional condition (Nom): for each $i \in \nom$, if $V(i)= \{w \}$ and $V'(i)=\{ w' \}$, then $Z(w,w')$~\cite{areces200714}. 
We write $\M, w \leftrightarroweq_{\DD} \M', w'$
if there is a $\DD$-bisimulation $Z$ between $\M$  and $\M'$ such that $(w,w') \in Z$.
Similarly, in the cases of \MLC{} and $\HL(@)$ we write $\M, w \leftrightarroweq_{\MLC} \M', w'$ and $\M, w \leftrightarroweq_\HL \M', w'$, respectively.
These bisimulations satisfy invariance lemmas for the corresponding languages, namely
if $\M, w \leftrightarroweq_{\MLC} \M', w'$ (resp. $\M, w \leftrightarroweq_\HL \M', w'$), then, for any $\MLC$-formula (resp. $\HL(@)$-formula) $\varphi$, it holds that $\M, w \models \varphi$ if and only if $\M', w' \models \varphi$ \cite{areces2010modal,areces200714}. 
Next, we show an analogous result for $\DD$-bisimulation.

\begin{lemma}\label{bisim}
If $\M, w  \leftrightarroweq_\DD \M', w'$ then, for any $\MLi$-formula $\varphi$, it holds that $\M, w \models \varphi$ if and only if $\M', w' \models \varphi$.
\end{lemma}
\begin{proof}
Assume that  $Z$ is a $\DD$-bisimulation between models $\M=(W,R,V)$ and ${\M'=(W',R',V')}$ satisfying $\M, w  \leftrightarroweq_\DD \M', w'$.
The proof is by induction on the structure of $\varphi$, where the
non-standard part is for the inductive step for \DD{}s, where $\varphi$ is of the form $@_{ \psi_1} \psi_2$.
If
$\M ,w \models @_{ \psi_1} \psi_2$, there is  a unique world $v \in W$ such that $\M, v \models \psi_1$, and moreover $\M, v \models \psi_2$.
As $Z$ is serial,  there is  $v' \in Z(v)$, and so, 
by the inductive assumption, $\M' , v' \models \psi_1 \land \psi_2$.
Suppose towards a contradiction that
$\M' ,w' \not\models @_{ \psi_1} \psi_2$, so
there is  $u' \neq v'$ such that $\M' , u' \models \psi_1$.
Since $Z$ is surjective,  there is $u \in W$ such that $u' \in Z(u)$.
Moreover, by the inductive assumption we obtain that $\M,u \models \psi_1$.
However, $v$ is the only world in $W$ which satisfies $\psi_1$, so $u=v$ and consequently $u' \in Z(v)$.
For the same reason there cannot be in $W$ any  world different than $v$ which is mapped by $Z$ to $v'$.
Hence, $Z^{-1}(v') = \{v \}$ and thus $Z(v) = \{ v' \}$.
This, however, contradicts the fact that $u' \in Z(v)$ and $u' \neq v'$.
The opposite implication is shown analogously. 
\qed
\end{proof}

We will exploit bisimulations in our analysis.
We start by considering arbitrary frames and we show that
$
\HL(@) \prec \MLi 
$
and $\MLi \prec \MLC$.

\begin{theorem}\label{thm:expres_@_iota}
It holds that $\HL(@) \prec \MLi$;
the result holds already over the class of finite frames.
\end{theorem}
\begin{proof}
Given an $\HL(@)$-formula $\varphi$ we construct an $\MLi$-formula $\varphi'$ by setting
$
{\varphi' = \varphi  \land \bigwedge_{i \in \nom(\varphi)} @_{ i} \top }.
$
The conjunction $\bigwedge_{i \in \nom(\varphi)} @_{i} \top$ guarantees that each $i \in \nom(\varphi)$ holds in exactly one world,
so 
$\HL(@) \preccurlyeq \MLi$.

To prove that $\MLi \not\preccurlyeq  \HL(@)$, 
we show that the $\MLi$-formula
$@_{ \top} \top$,
defining the class of frames with exactly one world, cannot be expressed in $\HL(@)$.
For this, we
construct models $\M$ and $\M'$ and an $\HL$-bisimulation $Z$ between them:

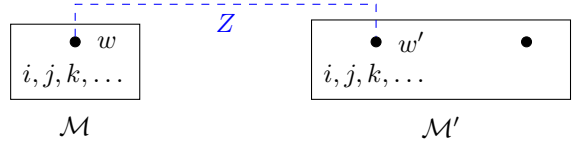
\begin{center}
\begin{tikzpicture}

\tikzset{>=latex}

\node[draw,circle,fill,minimum size=0.1pt,scale=0.4] (w1) at (0,0) {};
\node[draw=none,right=0.1 of w1] (t1) {$w$};
\node[draw=none,below=0.1 of w1] (t1') {$i,j,k, \dots$};
\node[draw=none,scale=0.4] (blank) at (0.5,0) {};

\node[draw,circle,fill,minimum size=0.1pt,scale=0.4] (w2) at (4,0) {};
\node[draw=none,right=0.1 of w2] (t2) {$w'$};
\node[draw=none,below=0.1 of w2] (t2') {$i,j,k, \dots$};
\node[draw,circle,fill,minimum size=0.1pt,scale=0.4] (w3) at (6,0) {};
\node[draw=none,scale=0.4] (blank') at (6.5,0) {};

\node[fit=(t1)(t1')(blank), draw, inner sep=1pt] (M) {};
\node[draw=none,below=0.1 of M] {$\M$};

\node[fit=(t2)(t2')(w3)(blank'), draw, inner sep=1pt] (M') {};
\node[draw=none,below=0.1 of M'] {$\M'$};

\draw [-, dashed, blue] (w1) --  (0,0.5) -- node[below] {$Z$}  (4,0.5) -- (w2);

\end{tikzpicture}
\end{center}

\noindent 
Clearly 
$\M,w \models @_{ \top} \top$, but $\M',w' \not\models @_{ \top} \top$.
However, since  $Z$  is an $\HL$-bisimulation,  
there exists no   $\HL(@)$-formula which holds in $w$, but not in $w'$.
\qed
\end{proof}



Next, we  use $\DD$-bisimulation to show that $\MLi \prec \MLC$.

\begin{theorem}\label{thm:expres_count_iota}
It holds that $\MLi \prec \MLC$;
the result holds already over the class of finite frames.
\end{theorem}
\begin{proof}
To show that $\MLi \preccurlyeq \MLC$,
we observe that $@_{ \varphi} \psi$ can be expressed as 
$\E (\varphi \land \psi \land \neg \D \varphi)$, where $\E$ and $\D$ are the `somewhere' and `difference' operators.
Both $\E$ and $\D$ can be expressed in \MLC{}, for example, 
$\E \varphi$ can be expressed as $\exists_{\geq 1} \varphi$ and $\D \varphi$ as $(\varphi \to \exists_{\geq 2} \varphi) \land (\neg \varphi \to \exists_{\geq 1} \varphi)$ \cite{areces2010modal}.
Thus $\MLi \preccurlyeq \MLC$.

To prove that $\MLC \not\preccurlyeq \MLi$, we 
show that 
$\MLi$ cannot express the $\MLC$-formula $\exists_{=2} \top$ defining frames with exactly two worlds in the domain.
Indeed, consider models $\M$ and $\M'$  and  a $\DD$-bisimulation between them as below:

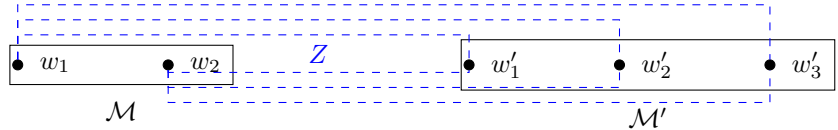
\begin{center}
\begin{tikzpicture}

\tikzset{>=latex}

\node[draw,circle,fill,minimum size=0.1pt,scale=0.4] (w1) at (0,0) {};
\node[draw=none,right=0.1 of w1] (t1) {$w_1$};
\node[draw,circle,fill,minimum size=0.1pt,scale=0.4] (w2) at (2,0) {};
\node[draw=none,right=0.1 of w2] (t2) {$w_2$};

\node[draw,circle,fill,minimum size=0.1pt,scale=0.4] (w3) at (6,0) {};
\node[draw=none,right=0.1 of w3] (t3) {$w_1'$};
\node[draw,circle,fill,minimum size=0.1pt,scale=0.4] (w4) at (8,0) {};
\node[draw=none,right=0.1 of w4] (t4) {$w_2'$};
\node[draw,circle,fill,minimum size=0.1pt,scale=0.4] (w5) at (10,0) {};
\node[draw=none,right=0.1 of w5] (t5) {$w_3'$};


\node[fit=(w1)(t2), draw, inner sep=1pt] (M) {};
\node[draw=none,below=0.1 of M] {$\M$};

\node[fit=(w3)(t5), draw, inner sep=1pt] (M') {};
\node[draw=none,below=0.1 of M'] {$\M'$};

\draw [-, dashed, blue] (w1) --  (0,0.4) -- (6,0.4) -- (w3);
\draw [-, dashed, blue] (w1) --  (0,0.6) -- (8,0.6) -- (w4);
\draw [-, dashed, blue] (w1) --  (0,0.8) -- (10,0.8) -- (w5);

\draw [-, dashed, blue] (w2) --  (2,-0.1) --  node[above] {$Z$} (6,-0.1) -- (w3);
\draw [-, dashed, blue] (w2) --  (2,-0.3) -- (8,-0.3) -- (w4);
\draw [-, dashed, blue] (w2) --  (2,-0.5) -- (10,-0.5) -- (w5);

\end{tikzpicture}
\end{center}

\noindent
Clearly $\M ,w_1 \models \exists_{=2} \top$, but $\M' ,w_1' \not\models \exists_{=2} \top$.
Since $Z$  is a $\DD$-bisimulation  mapping $w_1$ to $w_1'$,  these words satisfy the same $\MLi$-formulas.
\qed
\end{proof}

We note that the argument from the proof above, showing that there is no $\MLi$ formula which defines the class of frames with domains of cardinality $2$,
can be easily generalised to any cardinality larger than 1.
In contrast, as we showed in the proof of \Cref{thm:expres_@_iota}, the frame property of having the domain of cardinality 1 can 
be captured by the $\MLi$-formula  $@_{ \top} \top$.
In other words, $\MLi$ cannot define frames bigger than singletons.

Next, we focus on linear frames where the following result holds




\begin{theorem}\label{thm::overL}
The following relations hold: $\HL(@) \prec_L \MLi  \prec_L \MLC$.
\end{theorem}
\begin{proof}
Clearly, $\HL(@) \preccurlyeq_L \MLi$
and
$\MLi  \preccurlyeq_L \MLC$
follow from  \Cref{thm:expres_@_iota,thm:expres_count_iota}, so it remains to show that
$\MLi \not\preccurlyeq_L \HL(@)$
and 
$\MLC \not\preccurlyeq_L \MLi$.

To show that $\MLi \not\preccurlyeq_L \HL(@)$
we construct  models $\M$ and $\M'$ over $\Z$  with an $\HL$-bisimulation $Z$, as depicted below (note that the accessibility relation in the models is the transitive closure of the relation depicted by arrows):

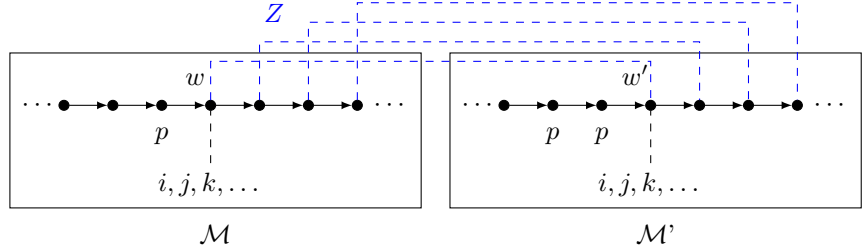
\begin{center}
\begin{tikzpicture}
[scale=0.65]

\tikzset{>=latex}

\node (d1) at (-3.5,0)  {$\cdots$};

\node[draw,circle,fill,minimum size=0.1pt,scale=0.4] (w-3) at (-3,0) {};
\node[draw,circle,fill,minimum size=0.1pt,scale=0.4] (w-2) at (-2,0) {};
\node[draw,circle,fill,minimum size=0.1pt,scale=0.4] (w-1) at (-1,0) {};
\node[draw=none,below=0.1 of w-1] (t-1) {$p$};

\node[draw,circle,fill,minimum size=0.1pt,scale=0.4] (w0) at (0,0) {};
\node[draw=none,below=0.7 of w0] (t0) {$i,j,k, \dots$};
\draw[-, dashed] (t0) -- (w0);
\node[draw=none,above left=0.1 and -0.2 of w0] (w) {$w$\color{white}$'$};

\node[draw,circle,fill,minimum size=0.1pt,scale=0.4] (w1) at (1,0) {};
\node[draw,circle,fill,minimum size=0.1pt,scale=0.4] (w2) at (2,0) {};
\node[draw,circle,fill,minimum size=0.1pt,scale=0.4] (w3) at (3,0) {};

\node (d2) at (3.7,0)  {$\cdots$};

\draw[->] (w-3) -- (w-2);
\draw[->] (w-2) -- (w-1);
\draw[->] (w-1) -- (w0);
\draw[->] (w0) -- (w1);
\draw[->] (w1) -- (w2);
\draw[->] (w2) -- (w3);


\node (d1') at (5.5,0) {$\cdots$};

\node[draw,circle,fill,minimum size=0.1pt,scale=0.4] (w-3') at (6,0) {};
\node[draw,circle,fill,minimum size=0.1pt,scale=0.4] (w-2') at (7,0) {};
\node[draw=none,below=0.1 of w-2'] (t-2') {$p$};
\node[draw,circle,fill,minimum size=0.1pt,scale=0.4] (w-1') at (8,0) {};
\node[draw=none,below=0.1 of w-1'] (t-1') {$p$};

\node[draw,circle,fill,minimum size=0.1pt,scale=0.4] (w0') at (9,0) {};
\node[draw=none,below=0.7 of w0'] (t0') {$i,j,k, \dots$};
\draw[-, dashed] (t0') -- (w0');
\node[draw=none,above left=0.1 and -0.15 of w0'] (w') {$w'$};

\node[draw,circle,fill,minimum size=0.1pt,scale=0.4] (w1') at (10,0) {};
\node[draw,circle,fill,minimum size=0.1pt,scale=0.4] (w2') at (11,0) {};
\node[draw,circle,fill,minimum size=0.1pt,scale=0.4] (w3') at (12,0) {};

\node (d2') at (12.7,0) {$\cdots$};

\draw[->] (w-3') -- (w-2');
\draw[->] (w-2') -- (w-1');
\draw[->] (w-1') -- (w0');
\draw[->] (w0') -- (w1');
\draw[->] (w1') -- (w2');
\draw[->] (w2') -- (w3');

\node[fit=(d1)(d2)(t0)(w), draw, inner sep=1pt] (M) {};
\node[draw=none,below=0.1 of M] {$\M$};

\node[fit=(d1')(d2')(t0')(w'), draw, inner sep=1pt] (M') {};
\node[draw=none,below=0.1 of M'] {$\M$'};

\draw [-, dashed, blue] (w0) --  (0,0.9) -- (9,0.9) -- (w0');
\draw [-, dashed, blue] (w1) --  (1,1.3) -- (10,1.3) -- (w1');
\draw [-, dashed, blue] (w2) --  (2,1.7) -- (11,1.7) -- (w2');
\draw [-, dashed, blue] (w3) --  (3,2.1) -- (12,2.1) -- (w3');
\node[blue] at (1.3,1.9) {$Z$};

\end{tikzpicture}
\end{center}

\noindent
Clearly $\M,w \models @_{ p} \top$, but 
$\M',w' \not\models @_{ p} \top$.
However, since $Z$ is an $\HL$-bisimulation mapping $w$ to $w'$, these worlds need to satisfy the same $\HL(@)$-formulas.

To show that  $\MLC \not\preccurlyeq_L \MLi$
we construct models $\MN$ and $\MN'$, each of them over a frame $\Z + \Z$ consisting of two copies of $\Z$, as depicted below:

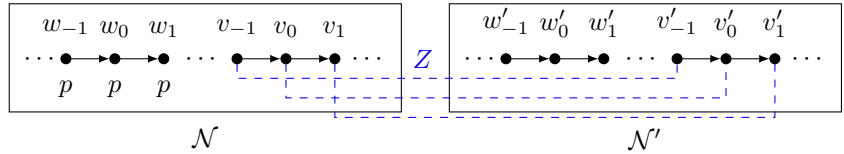
\begin{center}
\begin{tikzpicture}
[scale=0.65]

\tikzset{>=latex}

\node (d1) at (-0.5,0)  {$\cdots$};

\node[draw,circle,fill,minimum size=0.1pt,scale=0.4] (w0) at (0,0) {};
\node[draw=none,above=0.1 of w0] {$w_{-1}$};
\node[draw=none,below=0.1 of w0] (t0) {$p$};

\node[draw,circle,fill,minimum size=0.1pt,scale=0.4] (w1) at (1,0) {};
\node[draw=none,above=0.1 of w1]  {$w_{0}$};
\node[draw=none,below=0.1 of w1]  {$p$};

\node[draw,circle,fill,minimum size=0.1pt,scale=0.4] (w2) at (2,0) {};
\node[draw=none,above=0.1 of w2] {$w_{1}$};
\node[draw=none,below=0.1 of w2]  {$p$};

\node (d2) at (2.8,0)  {$\cdots$};

\node[draw,circle,fill,minimum size=0.1pt,scale=0.4] (v0) at (3.5,0) {};
\node[draw=none,above=0.1 of v0] {$v_{-1}$};

\node[draw,circle,fill,minimum size=0.1pt,scale=0.4] (v1) at (4.5,0) {};
\node[draw=none,above=0.1 of v1] {$v_{0}$};

\node[draw,circle,fill,minimum size=0.1pt,scale=0.4] (v2) at (5.5,0) {};
\node[draw=none,above=0.1 of v2]  {$v_{1}$};

\node (d3) at (6.2,0) {$\cdots$};

\draw[->] (w0) -- (w1);
\draw[->] (w1) -- (w2);
\draw[->] (v0) -- (v1);
\draw[->] (v1) -- (v2);

\node (d1') at (8.5,0)  {$\cdots$};

\node[draw,circle,fill,minimum size=0.1pt,scale=0.4] (w0') at (9,0) {};
\node[draw=none,above=0.1 of w0']  {$w'_{-1}$};
\node[draw=none,below=0.1 of w0'] (t0') {\color{white}$p$};

\node[draw,circle,fill,minimum size=0.1pt,scale=0.4] (w1') at (10,0) {};
\node[draw=none,above=0.1 of w1'] {$w'_{0}$};

\node[draw,circle,fill,minimum size=0.1pt,scale=0.4] (w2') at (11,0) {};
\node[draw=none,above=0.1 of w2']  {$w'_{1}$};

\node (d2') at (11.8,0)  {$\cdots$};

\node[draw,circle,fill,minimum size=0.1pt,scale=0.4] (v0') at (12.5,0) {};
\node[draw=none,above=0.1 of v0'] {$v'_{-1}$};

\node[draw,circle,fill,minimum size=0.1pt,scale=0.4] (v1') at (13.5,0) {};
\node[draw=none,above=0.1 of v1'] {$v'_{0}$};

\node[draw,circle,fill,minimum size=0.1pt,scale=0.4] (v2') at (14.5,0) {};
\node[draw=none,above=0.1 of v2'] {$v'_{1}$};

\node (d3') at (15.2,0) {$\cdots$};

\draw[->] (w0') -- (w1');
\draw[->] (w1') -- (w2');
\draw[->] (v0') -- (v1');
\draw[->] (v1') -- (v2');

\node[fit=(d1)(d3)(t0)(w), draw, inner sep=2pt] (M) {};
\node[draw=none,below=0.1 of M] {$\MN$};

\node[fit=(d1')(d3')(t0')(w'), draw, inner sep=2pt] (M') {};
\node[draw=none,below=0.1 of M'] {$\MN'$};


\draw [-, dashed, blue] (v0) --  (3.5,-0.4) -- (12.5,-0.4) -- (v0');
\draw [-, dashed, blue] (v1) --  (4.5,-0.8) -- (13.5,-0.8) -- (v1');
\draw [-, dashed, blue] (v2) --  (5.5,-1.2) -- (14.5,-1.2) -- (v2');

\node[blue] at (7.3,0) {$Z$};

\end{tikzpicture}
\end{center}

\noindent 
It holds that $\MN,v_0 \models \exists_{\geq 1} p$, but 
$\MN',v_0' \not\models \exists_{\geq 1} p$.
However, we can show that $v_0$ and $v_0'$ satisfy the same \MLi{}-formulas.
To this end, we observe that $Z$ is a (standard) bisimulation, so $v_0$ and $v_0'$ satisfy the same formulas from the basic modal language.
The language of $\MLi$ contains also formulas of the form $@_{  \varphi} \psi$, but  none of them is satisfied in any world of $\MN$ or $\MN'$.
Indeed, in the case of $\MN$ we can construct a $\DD$-bisimulation $Z_\MN$ between $\MN$ and itself which consists of  pairs $(w_n,w_m)$ and $(v_n,v_m)$ for all $n,m \in Z$.
Hence, all worlds of the form $w_n$ satisfy the same $\MLi$-formulas, and the same holds for all worlds $v_n$.
Thus, no formula of the form $@_{  \varphi} \psi$ can be satisfied in $\MN$, as there are either no worlds satisfying $\varphi$ or there are infinitely many of them.
An analogous argument shows that  no formula of the form $@_{  \varphi} \psi$ can be satisfied in $\MN'$.
\qed
\end{proof}

Next we show that expressiveness results change when we consider frames  $\Z$.

\begin{theorem}\label{thm::expsZ}
The following relations hold: $\HL(@) \prec_\Z \MLi \approx_\Z \MLC$.
\end{theorem}
\begin{proof}
The fact that  $\HL(@) \prec_\Z \MLi$ follows from the proof of \Cref{thm::overL} as the $\HL$-bisimulation constructed therein is over $\Z$. 
To show  $\MLi \approx_\Z \MLC$ it suffices to prove  $\MLC   \preccurlyeq_\Z \MLi$, as
$\MLi  \preccurlyeq_\Z \MLC$ follows  from \Cref{thm:expres_count_iota}.

To express $\MLC$-formulas in $\MLi$ it will be convenient to introduce, for any $n \in \N$, a formula
$\psi_n$ as  the following abbreviation
$$
\psi_n = \psi \land \Diamond ( \psi \land \Diamond ( \psi \land \dots ) ), \qquad \text{where $\psi$ occurs $n$ times}.
$$ 
We observe that  by the irreflexivity of the accessibility relation over $\Z$ we obtain that 
$\psi_n$ holds 
in all worlds $w_1$ of a model such that there exists a chain ${ w_1 < w_2 < \dots < w_n }$ of (not necesarily consecutive) distinct worlds satisfying $\psi$.

Given an $\MLC$-formula $\varphi$, we let $\varphi'$ be an $\MLi$-formula  obtained by
replacing in $\varphi$ each 
$\exists_{\geq n} \psi$ with $ \Diamond \psi_n \lor @_{ (\psi_n \land  \neg \Diamond \psi_n)} \top $.
To show that $\varphi$ and $\varphi'$ are equivalent over $\Z$
it suffices to show that 
$\exists_{\geq n} \psi$ is equivalent to $ \Diamond \psi_n \lor @_{ (\psi_n \land  \neg \Diamond \psi_n)} \top $.
Indeed, $\exists_{\geq n} \psi$ holds at $w$ if 
either 
(1)
there are $w_1 < \dots < w_n$, all larger than $w$, in which $\psi$ holds
or
(2) there exists the unique $w'$ such that $\psi$ holds in $w'$ and in exactly $n-1$ words larger than $w'$.
The first condition is expressed by $ \Diamond \psi_n$ and the second by $@_{ (\psi_n \land  \neg \Diamond \psi_n)} \top $, so $\exists_{\geq n} \psi$ is equivalent to $ \Diamond \psi_n \lor @_{ (\psi_n \land  \neg \Diamond \psi_n)} \top $.
Note that the disjunct $\Diamond \psi_n$ would  not be needed over finite linear frames.
\qed
\end{proof}

Observe that in the proof above we have shown that over $\Z$  \MLi{} allows us to count the number of occurrences of $p$ in a model, which is impossible over arbitrary frames and over linear frames, as we showed in the proof of \Cref{thm::overL}.










\section{Conclusions}\label{conclusions}

In this paper we have studied the computational complexity and expressive power of modal operators for definite descriptions.
We have shown that after adding Boolean \DD{}s to the basic modal language the satisfiability problem remains \PS{}-complete, so such an extension can be obtained with no computational cost.
However, if we allow for arbitrary \DD{}s, the problem becomes \EXPT{}-complete,
so the computational price is the same as for adding the universal modal operator or counting quantifiers with numbers encoded in unary.
Moreover, we have shown that in this setting \DD{}s provide strictly higher expressive power than the (basic) hybrid machinery, but strictly lower expressive power than counting operators.
The same holds over linear structures, but over integers \DD{}s become as expressive as counting operators.

Regarding the future research directions, it would be interesting to provide
a complexity-wise optimal decision procedure for \MLi{}-satisfiability, for example, using a tableaux systems.
We would also like 
 to study the complexity and expressiveness of well-behaving fragments of modal logic, such as Horn fragments.







\section*{Acknowledgments}
This research is funded by the European Union (ERC, ExtenDD, project number: 101054714). Views and opinions expressed are however those of the authors only and do not necessarily reflect those of the European Union or the European Research Council. Neither the European Union nor the granting authority can be held responsible for them.

\bibliographystyle{splncs04}
\bibliography{DDbiblio}

\end{document}